\documentclass[11pt]{article}
\usepackage[utf8]{inputenc}
\usepackage[british]{babel}
\usepackage{lmodern}
\usepackage{comment}
\usepackage{amsthm,amssymb,amsmath}
\usepackage{mathtools}
\usepackage{graphicx,algorithmic,algorithm}%
\usepackage{enumerate}
\usepackage{thmtools}
\usepackage{thm-restate}

\usepackage{tikz,tikz-cd}
\usepackage{circuitikz}
\tikzcdset{scale cd/.style={every label/.append style={scale=#1},
    cells={nodes={scale=#1}}}}%
\usetikzlibrary{decorations.markings}
\usetikzlibrary{shapes}
\usepackage[letterpaper, margin=1in]{geometry}
\usepackage{enumitem}
\usepackage{hyperref}
\usepackage{cleveref}

\usepackage{todonotes}
\usepackage[normalem]{ulem}
\hypersetup{
	pdftitle=   {},
	pdfauthor=  {}
}
\usepackage{authblk}

\usepackage{caption}
\usepackage{subcaption}

\newtheorem{theorem}{Theorem}[section]
\newtheorem{lemma}[theorem]{Lemma}

\newtheorem{claim}[theorem]{Claim}

\newenvironment{clproof}{\begin{list}{}{
			\setlength{\leftmargin}{3mm}
		} \item {\it Proof.} }{\hfill$\lozenge$\end{list}}

\newcommand\extrafootertext[1]{
    \bgroup
    \renewcommand\thefootnote{\fnsymbol{footnote}}
    \renewcommand\thempfootnote{\fnsymbol{mpfootnote}}
    \footnotetext[0]{#1}
    \egroup
}

\newcommand\lab{\operatorname{lab}}
\newcommand\full{\operatorname{full}}

\newcommand\CW{\operatorname{CW}}
\newcommand\NLC{\operatorname{NLC}}

\author[1]{Shinwoo An} 
\author[2]{Yeonsu Chang}
\author[3]{Kyungjin Cho}
\author[2,4]{O-joung Kwon\thanks{Corresponding author}}
\author[2]{Myounghwan Lee}
\author[5]{Eunjin Oh}
\author[5]{Hyeonjun Shin}
\affil[1]{Department of Computer Science, Bar-Ilan University, Ramat Gan,~Israel}
\affil[2]{Department of Mathematics, Hanyang University, Seoul,~South~Korea}
\affil[3]{Institute of Algorithms \& Theory, Graz University of Technology, Graz,~Austria}
\affil[4]{Discrete Mathematics Group, Institute~for~Basic~Science~(IBS), Daejeon,~South~Korea}
\affil[5]{Department of Computer Science and Engineering, POSTECH, Pohang,~South~Korea}

\title{Pre-assignment problem for unique minimum vertex cover on \\ bounded clique-width graphs}

\begin{document}

\maketitle
\extrafootertext{
An extended abstract of this paper appeared in the proceedings of AAAI 2025~\cite{AnCCKLOS2025}.
Y. Chang, O. Kwon, and M. Lee are supported by the National Research Foundation of Korea (NRF) grant funded by the Ministry of Science and ICT (No. NRF-2021K2A9A2A11101617 and No. RS-2023-00211670). O. Kwon is also supported by the National Research Foundation of Korea (NRF) grant funded by Institute for Basic Science (IBS-R029-C1).
E. Oh and H.Shin are supported by Institute of Information \& Communications Technology Planning \& Evaluation (IITP) grant funded by the Korea government (MSIT) (No. RS-2024-00440239, Sublinear Scalable Algorithms for Large-Scale Data Analysis) and the National Research Foundation of Korea (NRF) grant funded by the Korea government (MSIT) (No. RS-2024-00358505).
}

	\extrafootertext{E-mail addresses: 
 \texttt{shinwoo.an@biu.ac.il} (S. An), \texttt{yeonsu@hanyang.ac.kr} (Y. Chang), 
 \texttt{kyungjin.cho@tugraz.at} (K. Cho),  \texttt{ojoungkwon@hanyang.ac.kr} (O. Kwon),
 \texttt{sycuel@hanyang.ac.kr} (M. Lee),
 \texttt{eunjin.oh@postech.ac.kr} (E. Oh), and \texttt{tlsguswns119@postech.ac.kr}
 (H. Shin)  }
\begin{abstract}
Horiyama et al.\ (AAAI 2024) considered the problem of generating instances with a unique minimum vertex cover under certain conditions. The \textsc{Minimum Pre-assignment for Uniquification of Minimum Vertex Cover} problem (shortly \textsc{Min PAU-VC}) is the problem, for given a graph $G$, to find a minimum set $S$ of vertices in $G$ such that among all minimum vertex covers of $G$, exactly one contains $S$. We show that \textsc{Min PAU-VC} is fixed-parameter tractable parameterized by clique-width, which improves an exponential algorithm for trees given by Horiyama et al. Among natural graph classes with unbounded clique-width, we show that the problem can be solved in linear time on split graphs and unit interval graphs.
\end{abstract}

\section{Introduction}

Designing AI algorithms to tackle NP-hard graph problems has become a prominent trend in the field of artificial intelligence. The inherent complexity of NP-complete problems presents a significant challenge, making them an ideal testbed for AI-driven approaches that aim to push the boundaries of what can be achieved in terms of efficiency and scalability. 
To evaluate the performance of those AI algorithms, it is essential to have robust benchmark datasets. Such datasets provide a controlled environment where the strengths and weaknesses of different algorithms can be systematically analyzed. 
As constructing a benchmark dataset is a critical aspect of AI research, 
several well-known benchmark datasets were presented such as TSPLIB, UCI, SATLIB, and DIMACS for various NP-hard combinatorial problems~\cite{reinelt1991tsplib,asuncion2007uci,hoos2000satlib}.

However, it seems hard to use them to evaluate the performances of AI algorithms for the \emph{uniqueness version} of combinatorial problems where a solution is unique. 
In several problems, the presence of a unique solution can lead to more efficient algorithms~\cite{THOMASON1978259,gabow1999unique}. Also, algorithms for the unique SAT problem are used as subroutines for its search version~\cite{scheder2017ppsz,hertli20143}.
Due to these reasons, the uniqueness version also has been extensively studied from both theory and practice~\cite{calabro2008complexity,hertli2014breaking}.
Therefore, generating graphs with a unique solution offers a valuable addition to benchmark datasets, enabling a more thorough evaluation of AI-driven solvers for the uniqueness version of combinatorial problems.

One natural approach for generating graphs with a unique solution is to make use of graphs in well-known benchmark datasets. More specifically, we choose a graph $G$ in a well-known benchmark dataset, and pre-assign a part of $G$ so that only one solution is consistent with this assignment. This \emph{pre-assignment for uniquification} has been studied for classic NP-hard problems such as the coloring and clique problems~\cite{harary2007computational}, the dominating set problem and its variants~\cite{MR1480799,bozeman2019restricted,ferrero2018relationship} and the vertex cover problem~\cite{horiyama2024theoretical}.
Also, several pencil/video puzzles such as SUDOKU and Picross 3D have been studied in the context of the pre-assignment for uniquification~\cite{demaine2018fewest,kimura_et_al:LIPIcs.FUN.2018.25,TjusilaT2024}.

\medskip 
In this paper, we focus on the \textsc{Pre-assignment for uniquification of Minimum Vertex Cover (PAU-VC)} problem
introduced by Horiyama et al.~\cite{horiyama2024theoretical}, and its optimization version. A set $S$ of vertices in a graph $G$ is called a \emph{vertex cover} of $G$ if $S$ meets all edges of $G$. The formal definition of \textsc{PAU-VC} and \textsc{Min PAU-VC} are the following.

\medskip 
\noindent
\fbox{\parbox{\textwidth}{
	\textsc{PAU-VC}\\
	\textbf{Input :} A graph $G$ and an integer $k$ 
	\\
	\textbf{Question :} Is there a set $S$ of at most $k$ vertices in $G$ such that among all minimum vertex covers of $G$, exactly one contains $S$? }}
\medskip

\noindent
\fbox{\parbox{\textwidth}{
	\textsc{Min PAU-VC}\\
	\textbf{Input :} A graph $G$ 
	\\
	\textbf{Question :} Find a minimum set $S$ of vertices in $G$ such that among all minimum vertex covers of $G$, exactly one contains $S$. }}
\medskip

Note that Horiyama et al.~\cite{horiyama2024theoretical} originally considered three models: INCLUDE, EXCLUDE, and MIXED. The problem we focus on corresponds to the INCLUDE model.

Notice that one can use an algorithm for \textsc{Min PAU-VC} to generate a graph with a unique solution for the vertex cover problem. Consider an arbitrary graph $G$ (possibly from a known benchmark dataset), and compute an optimal solution $S$ for \textsc{Min PAU-VC} on $G$. Since
there is a unique minimum vertex cover of $G$ containing $S$, $G-S$ has a unique minimum vertex cover as well, where $G-S$ is the graph obtained from $G$ by removing all vertices of $S$ and their incident edges. 

\medskip 
Although the pre-assignment for uniquification of the dominating set problem and its variants has been studied extensively, little is known about \textsc{PAU-VC}, except for~\cite{horiyama2024theoretical}. More specifically, 
Horiyama et al.~\cite{horiyama2024theoretical} proved that \textsc{PAU-VC} is $\Sigma_{2}^P$-complete on general graphs, and NP-complete on bipartite graphs. On the positive side, they provided an algorithm that runs in time $\mathcal{O}(2.1996^n)$ for general graphs, an algorithm that runs in time $\mathcal{O}(1.9181^n)$ for bipartite graphs, 
and an algorithm that runs in time $\mathcal{O}(1.4143^n)$ for trees,
where $n$ denotes the number of vertices.
Horiyama, Seto, and Suzuki~\cite{HoriyamaSS2025} further improved the running time for bipartite graphs to $\mathcal{O}(1.4143^n)$, and also showed that the problem is NP-complete on planar bipartite graphs of maximum degree $3$.

As \textsc{PAU-VC} is $\Sigma_2^P$-complete and NP-complete for general graphs and bipartite graphs, respectively, it is unlikely to admit polynomial-time algorithms for either general or bipartite graphs. However, the time complexity for trees remains an open question. 
In fact, Horiyama et al.~\cite{horiyama2024theoretical} also mentioned this explicitly: \emph{``Many readers might consider that \textsc{PAU-VC} for trees is likely solvable in polynomial time. On the other hand, not a few problems are intractable (e.g., \textsc{Node Kayles}) in general, but the time complexity for trees still remains open, and only exponential-time algorithms are known. In the case of \textsc{PAU-VC}, no polynomial-time algorithm for trees is currently known.''}

\medskip
In this paper, we resolve this open problem by presenting a polynomial-time algorithm for  \textsc{PAU-VC} on trees, which significantly improves the exponential-time algorithm by Horiyama et al. 
Moreover, we showed that it can be extended to classes of bounded \emph{clique-width}~\cite{CourcelleO2000}.  
Clique-width is a graph parameter that measures the complexity of constructing a graph using a set of specific operations, including the creation of new vertices, disjoint union of graphs, relabeling of vertex labels, and connecting vertices based on their labels. 
Trees have clique-width at most $3$~\cite{CourcelleO2000} and complete graphs have clique-width at most $2$. 
A precise definition will be given in Section~\ref{sec:prelim}. 

More precisely, we prove the following theorem. 
In parameterized complexity, an instance of a parameterized problem consists in a pair $(x,k)$, where $k$ is a secondary measurement, called the parameter. A parameterized problem $Q\subseteq \Sigma^*\times \mathbb{N}$ is \emph{fixed-parameter tractable (FPT)} if there is an algorithm which decides whether $(x,k)$ belongs to $Q$ in time $f(k)\cdot |x|^{\mathcal{O}(1)}$ for some computable function $f$.

\begin{theorem}\label{thm:maincliquewidth}
    \textsc{Min PAU-VC} is fixed-parameter tractable parameterized by clique-width.
\end{theorem}

The notion of clique-width is closely related to the concept of \emph{tree-width}. 
Tree-width is a well-studied graph parameter which measures how close a graph is to being a tree~\cite{GMXX}. Courcelle~\cite{Courcelle1990} showed that every problem expressible in $\text{MSO}_2$-logic is fixed-parameter tractable when parameterized by the tree-width of a graph. 
However, classes of bounded tree-width must be sparse. To address this limitation, Courcelle and Olariu~\cite{CourcelleO2000} introduced clique-width to extend properties of classes of bounded tree-width to dense graph classes, such as the class of complete graphs. 

Every class of bounded tree-width has bounded clique-width~\cite{CourcelleO2000,CorneilR2005}, but there are classes of bounded clique-width and unbounded tree-width, such as the class of complete graphs or complete bipartite graphs.
Courcelle, Makowsky, and Rotics~\cite{CourcelleMR2000} showed that every problem expressible in $\text{MSO}_1$-logic is fixed-parameter tractable when parameterized by the clique-width of a graph.

Although the property of being a vertex cover is expressible in $\text{MSO}_1$, expressing that a vertex cover is minimum would require comparing its cardinality with that of every other vertex cover. Such cardinality comparisons are not available in plain $\text{MSO}_1$. Thus, \textsc{PAU-VC} is not directly captured by a standard $\text{MSO}_1$-formulation, and
the algorithmic meta theorem by Courcelle, Makowsky, and Rotics cannot be adapted for \textsc{PAU-VC}. The parameterized complexity of problems cannot be expressible by $\text{MSO}_1$-logic, such as \textsc{Hamiltonian Cycle} and \text{Graph Coloring}, have been studied~\cite{KoblerR2003,Fomin2010,Fomin2014,Fomin2019,BergougnouxKK2020}.

\medskip
One may ask whether we can further obtain polynomial-time algorithms for \textsc{Min PAU-VC} on natural classes of graphs of unbounded clique-width. We investigate two such classes. \emph{Split graphs} are graphs that can be partitioned into an independent set and a clique. \emph{Unit interval graphs} are intersection graphs of intervals of the same length on the real line. 
It is known that 
split graphs have unbounded clique-width~\cite{MakowskyR1999} and  unit interval graphs have unbounded clique-width~\cite{GolumbicR2000}. 
Split graphs and unit interval graphs are well-known graph classes that have been widely studied~\cite{corneil1995simple,hell2001fully,bertossi1984dominating}. 
We prove that \textsc{Min PAU-VC} can be solved in linear time on both classes. 

\begin{theorem}
    \textsc{Min PAU-VC} can be solved in linear time on unit interval graphs and split graphs.
\end{theorem}

Note that the class of split graphs and the class of unit interval graphs are well-known subclasses of the class of \emph{chordal graphs}. It would be interesting to determine whether \textsc{Min PAU-VC} can be solved in polynomial time on chordal graphs. 

   This paper is organized as follows. 
    In Section~\ref{sec:prelim}, we introduce basic definitions and notations, including clique-width and NLC-width.
    In Section~\ref{sec:cliquewidth}, we present a fixed-parameter tractable algorithm for \textsc{Min PAU-VC} parameterized by clique-width. We present  linear time algorithms for \textsc{Min PAU-VC} on unit interval graphs in Section~\ref{sec:unitinterval} and on split graphs in Section~\ref{sec:split}. We discuss some open problems in Section~\ref{sec:conclusion}.

\section{Preliminary}\label{sec:prelim}

For every positive integer $n$, let $[n]$ denote the set of positive integers at most $n$.
All graphs in this paper are simple and finite. 
For a graph $G$, we denote by $V(G)$ and $E(G)$ the vertex set and edge set of $G$, respectively.
For graphs $G$ and $H$, let $G\cup H$ be the graph with vertex set $V(G)\cup V(H)$ and edge set $E(G)\cup E(H)$. 

Let $G$ be a graph. For a vertex $v$ of a graph $G$, let $N_G(v)$ denote the set of \emph{neighbors} of $v$ in $G$.
For $X\subseteq V(G)$, let $G[X]$ denote the subgraph of $G$ induced by $X$.
We denote by $G - X$ the graph $G[V(G)\setminus X]$, and for a single vertex $x \in V(G)$, we use the shorthand `$G - x$' for `$G - \{x\}$'. 
For two sets $X,Y\subseteq V(G)$, let $G[X,Y]$ be the graph $(X \cup Y, \{xy \in E(G) : x \in X, y \in Y\})$.

A set $X \subseteq V(G)$ is a \emph{clique} if any two vertices of $X$ are adjacent in $G$,
and it is an \emph{independent set} if any two vertices of $X$ are not adjacent in $G$.

\subsection{Clique-Width and NLC-Width}
Let $k$ be a positive integer. 
A \emph{$k$-labeled graph} is a pair $(G, \lab_G)$ of a graph $G$ and a function $\lab_G:V(G)\to [k]$, called the \emph{labeling function}. We denote by $\lab_G^{-1}(i)$ the set of vertices in $G$ with label $i$.

We first define the \emph{clique-width} of graphs. 
For a $k$-labeled graph $(G, \lab_G)$ and $i,j\in [k]$ with $i\neq j$, let $\eta_{i,j}(G, \lab_G)$ be the $k$-labeled graph obtained from $(G, \lab_G)$ by adding an edge between every vertex of label $i$ and every vertex of label $j$, and let $\rho_{i\to j}(G, \lab_G)$ be the $k$-labeled graph obtained from $(G, \lab_G)$ by relabeling every vertex of $i$ to $j$. For two vertex-disjoint $k$-labeled graphs $(G, \lab_G)$ and $(H, \lab_H)$, let $(G, \lab_G)\oplus (H, \lab_H)$ be the disjoint union of them.

The class $\CW_k$ of $k$-labeled graphs is recursively defined as follows.
\begin{itemize}
    \item The single vertex graph $i(x)$, with a vertex $x$ labeled with $i\in [k]$, is in $\CW_k$.
    \item Let $(G, \lab_G)$ and $(H, \lab_H)$ be two vertex-disjoint $k$-labeled graphs in $\CW_k$.
    Then $(G, \lab_G)\oplus (H, \lab_H)\in \CW_k$.
    \item Let $(G, \lab_G)\in \CW_k$ and $i,j\in [k]$ with $i\neq j$. Then 
    $\eta_{i,j}(G, \lab_G)\in \CW_k$.
    \item  Let $(G, \lab_G)\in \CW_k$ and $i,j\in [k]$. Then 
     $\rho_{i\to j}(G, \lab_G)\in \CW_k$.
    \end{itemize}
A \emph{clique-width $k$-expression} is a finite term built with the four operations above and using at most $k$ labels. The \emph{clique-width} of a graph is the minimum $k$ such that $(G, \lab_G)\in \CW_k$ for some labeling $\lab_G$. 

For example, 
\[ \eta_{2,3}\Big( \eta_{1,2}\big(1(a)\oplus 2(b)\big) \oplus \eta_{1,3}\big(3(c)\oplus 1(d)\big) \Big) \]
is a clique-width $3$-expression of a path $P_4$ on $4$ vertices. See Figure~\ref{fig:clique_width}. Thus, $P_4$ has clique-width at most $3$.
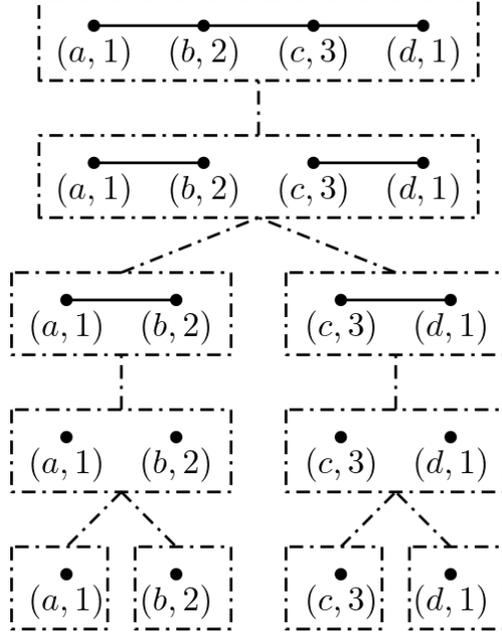
\begin{figure}[t]
\centering
\resizebox{0.85\linewidth}{!}{
\begin{tikzpicture}[
    gbox/.style={draw, rectangle, inner sep=2pt, rounded corners=2pt, fill=white},
    gnode/.style={circle, fill=black, inner sep=1.5pt},
    thick
]
    \node[gbox, label=left:{$\eta_{2,3}$}] (root) at (0, 0) {
        \begin{tikzpicture}[x=0.5cm, y=0.5cm]
            \node[gnode, label=below:{\scriptsize $1$}, label=above:{\scriptsize $a$}] (a) at (0,0) {};
            \node[gnode, label=below:{\scriptsize $2$}, label=above:{\scriptsize $b$}] (b) at (1,0) {};
            \node[gnode, label=below:{\scriptsize $3$}, label=above:{\scriptsize $c$}] (c) at (2,0) {};
            \node[gnode, label=below:{\scriptsize $1$}, label=above:{\scriptsize $d$}] (d) at (3,0) {};
            \draw (a) -- (b) -- (c) -- (d);
        \end{tikzpicture}
    };

    \node[gbox, label=left:{$\oplus$}] (oplus_mid) at (0, -1.25) {
        \begin{tikzpicture}[x=0.5cm, y=0.5cm]
            \node[gnode, label=below:{\scriptsize $1$}, label=above:{\scriptsize $a$}] (a) at (0,0) {};
            \node[gnode, label=below:{\scriptsize $2$}, label=above:{\scriptsize $b$}] (b) at (1,0) {};
            \node[gnode, label=below:{\scriptsize $3$}, label=above:{\scriptsize $c$}] (c) at (2,0) {};
            \node[gnode, label=below:{\scriptsize $1$}, label=above:{\scriptsize $d$}] (d) at (3,0) {};
            \draw (a) -- (b); 
            \draw (c) -- (d);
        \end{tikzpicture}
    };

    \node[gbox, label=left:{$\eta_{1,2}$}] (eta_L) at (-3.5, -2.5) {
        \begin{tikzpicture}[x=0.5cm, y=0.5cm]
            \node[gnode, label=below:{\scriptsize $1$}, label=above:{\scriptsize $a$}] (a) at (0,0) {};
            \node[gnode, label=below:{\scriptsize $2$}, label=above:{\scriptsize $b$}] (b) at (1,0) {};
            \draw (a) -- (b);
        \end{tikzpicture}
    };
    
    \node[gbox, label=right:{$\eta_{1,3}$}] (eta_R) at (3.5, -2.5) {
        \begin{tikzpicture}[x=0.5cm, y=0.5cm]
            \node[gnode, label=below:{\scriptsize $3$}, label=above:{\scriptsize $c$}] (c) at (0,0) {};
            \node[gnode, label=below:{\scriptsize $1$}, label=above:{\scriptsize $d$}] (d) at (1,0) {};
            \draw (c) -- (d);
        \end{tikzpicture}
    };

    \node[gbox, label=left:{$\oplus$}] (oplus_L) at (-3.5, -3.75) {
        \begin{tikzpicture}[x=0.5cm, y=0.5cm]
            \node[gnode, label=below:{\scriptsize $1$}, label=above:{\scriptsize $a$}] (a) at (0,0) {};
            \node[gnode, label=below:{\scriptsize $2$}, label=above:{\scriptsize $b$}] (b) at (1,0) {};
        \end{tikzpicture}
    };
    
    \node[gbox, label=right:{$\oplus$}] (oplus_R) at (3.5, -3.75) {
        \begin{tikzpicture}[x=0.5cm, y=0.5cm]
            \node[gnode, label=below:{\scriptsize $3$}, label=above:{\scriptsize $c$}] (c) at (0,0) {};
            \node[gnode, label=below:{\scriptsize $1$}, label=above:{\scriptsize $d$}] (d) at (1,0) {};
        \end{tikzpicture}
    };

    \node[gbox, label=left:{$1(a)$}] (leaf_a) at (-5, -5) {
        \begin{tikzpicture}[x=0.5cm, y=0.5cm] \node[gnode, label=below:{\scriptsize $1$}, label=above:{\scriptsize $a$}] (a) at (0,0) {}; \end{tikzpicture}
    };
    \node[gbox, label=right:{$2(b)$}] (leaf_b) at (-2, -5) {
        \begin{tikzpicture}[x=0.5cm, y=0.5cm] \node[gnode, label=below:{\scriptsize $2$}, label=above:{\scriptsize $b$}] (b) at (0,0) {}; \end{tikzpicture}
    };
    \node[gbox, label=left:{$3(c)$}] (leaf_c) at (2, -5) {
        \begin{tikzpicture}[x=0.5cm, y=0.5cm] \node[gnode, label=below:{\scriptsize $3$}, label=above:{\scriptsize $c$}] (c) at (0,0) {}; \end{tikzpicture}
    };
    \node[gbox, label=right:{$1(d)$}] (leaf_d) at (5, -5) {
        \begin{tikzpicture}[x=0.5cm, y=0.5cm] \node[gnode, label=below:{\scriptsize $1$}, label=above:{\scriptsize $d$}] (d) at (0,0) {}; \end{tikzpicture}
    };

    \draw (root) -- (oplus_mid);
    \draw (oplus_mid) -- (eta_L);
    \draw (oplus_mid) -- (eta_R);
    \draw (eta_L) -- (oplus_L);
    \draw (eta_R) -- (oplus_R);
    \draw (oplus_L) -- (leaf_a);
    \draw (oplus_L) -- (leaf_b);
    \draw (oplus_R) -- (leaf_c);
    \draw (oplus_R) -- (leaf_d);

\end{tikzpicture}
} 
\caption{An illustration of a clique-width $3$-expression of $P_4$.}
\label{fig:clique_width}
\end{figure}
   
Now, we define the \emph{NLC-width} of graphs introduced by Wanke~\cite{Wanke1994}.
For two vertex-disjoint $k$-labeled graphs $(G, \lab_G)$ and $(H, \lab_H)$ and a set $M\subseteq [k]^2$ of label pairs, we define 
$(G, \lab_G)\times_M (H, \lab_H)\coloneqq ((V',E'),\lab')$  where 
\begin{itemize}
    \item $V'=V(G)\cup V(H)$,
    \item $E'=E(G)\cup E(H)\cup \{uv:u\in V(G), v\in V(H), (\lab_G(u), \lab_H(v))\in M\}$,
    \item $\lab'(u)=\lab_G(u)$ if $u\in V(G)$ and $\lab'(u)=\lab_H(u)$ otherwise.
\end{itemize}
In other words, $(G, \lab_G)\times_M (H, \lab_H)$ is obtained from the disjoint union of $(G, \lab_G)$ and $(H, \lab_H)$ by, for every $(i,j)\in M$, adding all edges between vertices of label $i$ in $G$ and vertices of label $j$ in $H$. For a $k$-labeled graph $(G, \lab_G)$ and a function $R:[k]\to [k]$, let  $\rho_R(G, \lab_G)=(G, \lab')$ where $\lab'(u)=R(\lab_G(u))$ for all $u\in V(G)$.

 The class $\NLC_k$ of $k$-labeled graphs is recursively defined as follows.
\begin{enumerate}
    \item The single vertex graph $i(x)$, with a vertex $x$ labeled with $i\in [k]$, is in $\NLC_k$.
    \item Let $(G, \lab_G)\in \NLC_k$, and let $R:[k]\to [k]$ be a function. Then 
    $\rho_R(G, \lab)\in \NLC_k$.
    \item Let $(G, \lab_G)$ and $(H, \lab_H)$ be two vertex-disjoint labeled graphs in $\NLC_k$, and $M\subseteq [k]^2$.
    Then $(G, \lab_G)\times_M (H, \lab_H)\in \NLC_k$.
\end{enumerate}
An \emph{NLC-width $k$-expression} is a finite term built with the three operations above and using at most $k$ labels. The \emph{NLC-width} of a graph $G$ is the minimum $k$ such that $(G, \lab_G)\in \NLC_k$ for some labeling $\lab_G$.

\begin{theorem}[Johansson~\cite{Johansson1998}]\label{thm:transform}
Let $k$ be a positive integer.
Every graph of clique-width at most $k$ has NLC-width at most $k$, and one can in polynomial time transform a clique-width $k$-expression to an NLC-width $k$-expression.
\end{theorem}

We remark about algorithms to find a clique-width expression when it is not given. 
Fellows, Rosamond, Rotics, and Szeider~\cite{FellowsRRS2006} proved that computing clique-width is NP-hard.
Oum and Seymour~\cite{OumS2006} first obtained an approximation algorithm that computes a clique-width $(2^{3k+2}-1)$-expression of a given graph $G$ of clique-width at most $k$, which runs in time $\mathcal{O}(8^k n^{9} \log n)$.
Oum~\cite{Oum2009} later improved this by providing two algorithms; one is an algorithm that computes a clique-width $(8^{k}-1)$-expression in time $\mathcal{O}(g(k)\cdot n^{3})$ for some function $g$, and the other one is an algorithm that computes a clique-width $(2^{3k+2}-1)$-expression in time $\mathcal{O}(8^k n^4)$.
More recently, Fomin and Korhonen~\cite{FominK2024} devised an algorithm that computes a $(2^{2k+1}-1)$-expression in time $f(k)\cdot n^2$ for some function $f$, thereby breaking the cubic barrier of previous algorithms. Subsequently,  
Korhonen and Soko\l owski~\cite{KorhonenS2024} improved the running time to almost linear time. 
We may use one of these algorithms to produce a clique-width expression, when it is not given as input.

\section{Graphs of bounded clique-width}\label{sec:cliquewidth}
In this section, we prove Theorem~\ref{thm:maincliquewidth}.
Before diving into our main theorems, we present an idea for having a simpler polynomial time algorithm for \textsc{Min PAU-VC} on trees in Subsection~\ref{sec:trees}.
We present a fixed-parameter tractable algorithm for \textsc{Min PAU-VC} parameterized by clique-width in Subsection~\ref{sec:fptcliquewidth}.

\subsection{Trees}\label{sec:trees}
Let $G$ be a tree. We choose an arbitrary vertex as the root of $G$, call it $r$.
For each node $v\in V(G)$, we use $G_v$ to denote the subtree of $G$ rooted at $v$. 

For each vertex $v$, there are two types of vertex covers of $G_v$; one is a vertex cover of $G_v$ containing $v$ and the other is a vertex cover of $G_v$ not containing $v$. We want to find a set $S$ which forces the number of minimum vertex covers of each type to satisfy a certain condition. This naturally suggests the following definition. For a function $\beta:\{0,1\}\to \{0,1,2\}$, a set $S\subseteq V(G_v)$  is a \emph{$\beta$-set} in $G_v$ if the following hold:
\begin{itemize}
    \item If $\beta(0)\in \{0,1\}$, then there are exactly $\beta(0)$ minimum vertex covers of $G_v$ not containing $v$ and containing $S$. 
    \item If $\beta(0)=2$, then there are at least two minimum vertex covers of $G_v$ not containing $v$ and containing $S$.
    \item If $\beta(1)\in \{0,1\}$, then there are exactly $\beta(1)$ minimum vertex covers of $G_v$ containing $v$ and containing $S$. 
    \item If $\beta(1)=2$, then there are at least two minimum vertex covers of $G_v$ containing $v$ and containing $S$.
\end{itemize}
We will recursively compute a minimum $\beta$-set in $G_v$ for every possible function $\beta$ and every vertex $v\in V(G)$, if one exists.

It is not difficult to observe that if we have a minimum $\beta$-set of $G_r=G$ for every possible function $\beta$, then we can find an optimal solution of \textsc{Min PAU-VC}. That would be a minimum set among minimum $\beta$-sets of $G_r$ for which $\beta(0)+\beta(1)=1$.

Therefore, it suffices to recursively compute a minimum $\beta$-set of $G_v$ for every vertex $v\in V(G)$. The idea is straightforward. We need to propagate the information to children of $v$. Assume $\beta:\{0,1\}\to \{0,1,2\}$ is a given function. For example, if $\beta(0)=1$, then the $\beta$-set in $G_v$ should force a unique minimum vertex cover of $G_v$ not containing $v$. Then for each child $w$ of $v$, we have to determine a set forcing a unique minimum vertex cover of $G_w$ that contains $w$. This suggests how to split $\beta$ into functions $\beta_w$ for each child $w$, and we can find the corresponding $\beta$-set by taking the union of $\beta_w$-sets for children $w$ of $v$. 

This idea is generalized into graphs of bounded clique-width in the next subsection.  We will provide the dynamic programming algorithm and prove the correctness.

\subsection{FPT algorithm parameterized by clique-width}\label{sec:fptcliquewidth}
Let $(H, \lab_H)$ be a $k$-labeled graph. For a set $X$ of vertices in $H$, we denote by $\full_H(X)$ the set of integers $i\in [k]$ where $\lab_H^{-1}(i)\subseteq X$.
For $I\subseteq [k]$, a set $T\subseteq V(H)$ is a \emph{minimum vertex cover of $H$ with respect to $I$} if it is a minimum set among all vertex covers $X$ of $H$ with $\full_H(X)=I$. Note that $T$ is not necessarily a minimum vertex cover of $H$. Let $\mu_H(I)$ be the size of a minimum vertex cover of $H$ with respect to $I$. If there is no such set, then we define it to be~$\infty$.

Assume that $(F, \lab_F)=(G, \lab_G)\times_M (H, \lab_H)$ for some $k$-labeled graphs $(G, \lab_G)$, $(H, \lab_H)$, and $M\subseteq [k]^2$. Observe that for every $(i,j)\in M$, every vertex cover of $F$ either contains all vertices of $\lab_G^{-1}(i)$ or contains all vertices of $\lab_H^{-1}(j)$.  Thus, in each side, it is necessary to consider vertex covers that fully contain vertex sets of certain labels. This is the reason why we define minimum vertex covers with respect to $I\subseteq [k]$.

Now, to find sets $S\subseteq V(F)$ that force to have a unique minimum vertex cover, in each of $G$ and $H$, we need to know whether a set forces to have a unique minimum vertex cover with respect to some $I\subseteq [k]$. For each $I\subseteq [k]$, we need to distinguish three statuses: (1) $S$ does not force any minimum vertex cover with respect to $I$, (2) $S$ forces a unique minimum vertex cover with respect to $I$, or (3) $S$ forces at least two minimum vertex covers with respect to $I$. This property will be captured by the notion of characteristic, defined below.

A function $\beta:2^{[k]}\to \{0,1,2\}$ is the \emph{characteristic} of a set $S\subseteq V(H)$ in $H$ if for every $J\subseteq [k]$, 
\begin{itemize}
    \item  if $\beta(J)\in \{0,1\}$, then there are exactly $\beta(J)$ minimum vertex covers of $H$ with respect to $J$ and containing $S$, and 
    \item if $\beta(J)=2$, then there are at least two minimum vertex covers of $H$ with respect to $J$ and containing $S$.
\end{itemize}
Such a set $S\subseteq V(H)$ is called a \emph{$\beta$-set in $H$}. Let $\Pi(H)$ be the collection of functions $\beta:2^{[k]}\to \{0,1,2\}$ such that there is a $\beta$-set in $H$.

In the following lemma, we explain how we can solve \textsc{Min PAU-VC} on a $k$-labeled graph $H$ if we know the set $\Pi(H)$ and the function $\mu_H$ and a collection of minimum $\beta$-sets for $\beta\in \Pi(H)$. 
\begin{lemma}\label{lem:inclusionfinalstep}
    Let $k$ be a positive integer. Given a $k$-labeled graph $(G, \lab_G)$ with $\Pi(G)$, $\mu_G$ and a collection of minimum $\beta$-sets for $\beta\in \Pi(G)$, one can solve \textsc{Min PAU-VC} for $G$ in time $2^{\mathcal{O}(2^k)}|V(G)|$. 
\end{lemma}
\begin{proof}
Let $\mu=\min_{I\subseteq [k]} \mu_G(I)$, and $\Gamma=\{J\subseteq [k]: \mu_G(J)=\mu\}$.
Then $\mu$ is the size of a minimum vertex cover of $G$. 
We say that a function $\beta:2^{[k]}\to \{0,1,2\}$ is \emph{valid} if $\sum_{J\in \Gamma}\beta(J)=1$. A $\beta$-set with a valid function $\beta$ in $\Pi(G)$ is a set forcing a unique minimum vertex cover in $G$. Thus, the minimum $\beta$-set with a valid function $\beta$ in $\Pi(G)$ is a required solution for \textsc{Min PAU-VC}.
\end{proof}

By Lemma~\ref{lem:inclusionfinalstep}, it is sufficient to compute $\Pi(H)$ and $\mu_H$ and a collection of minimum $\beta$-sets. We will compute them in a bottom-up way, along a given NLC-width $k$-expression. 

In the next lemma, we describe how to merge information for $(G, \lab_G)$ and $(H, \lab_H)$ to obtain information for $(G, \lab_G)\times_M (H, \lab_H)$.

\begin{lemma}\label{lem:product}
    Let $k$ be a positive integer, and let $(G, \lab_G)$ and $(H, \lab_H)$ be vertex-disjoint $k$-labeled non-empty graphs. Let $M\subseteq [k]^2$ and let $(F, \lab_F)=(G, \lab_G)\times_M (H, \lab_H)$. 
    
    Given $\Pi(G), \Pi(H)$ and $\mu_G, \mu_H$ and a collection $\mathcal{I}_G$ of minimum $\beta$-sets for $\beta\in \Pi(G)$ and a collection $\mathcal{I}_H$ of minimum $\beta$-sets for $\beta\in \Pi(H)$, one can compute $\Pi(F), \mu_F$ and a collection $\mathcal{I}_F$ of minimum $\beta$-sets for $\beta\in \Pi(F)$ in time $2^{\mathcal{O}(2^k)}|V(F)|$.
\end{lemma}
\begin{proof}
We construct an auxiliary bipartite graph $Q$ with bipartition $(\{a_i:i\in [k]\}, \{b_i:i\in [k]\})$ such that $a_i$ is adjacent to $b_j$ if and only if $(i,j)\in M$.
Let $A=\{a_i:i\in [k]\}$, $B=\{b_i:i\in [k]\}$, and let $g:V(Q)\to [k]$ be a function such that $g(a_i)=g(b_i)=i$ for all $i\in [k]$. Let $\mathcal{Z}$ be the collection of all vertex covers of $Q$. Note that the number of all vertex covers of $Q$ is at most $2^{2k}$.

We first compute $\mu_F(I)$ for each $I\subseteq [k]$, which is the size of a minimum vertex cover of $F$ with respect to $I$. Note that for any $(i,j)\in M$, every vertex cover of $F$ contains either all vertices of $\lab_G^{-1}(i)$ or all vertices of $\lab_H^{-1}(j)$. We guess a vertex cover of $Q$ corresponding to parts whose all vertices are contained in a vertex cover of $F$.

We construct a function $\mu^*$ as below. 
\begin{itemize}
    \item Let $I\subseteq [k]$. For each $Z\in \mathcal{Z}$ with $(g(Z\cap A)\cap g(Z\cap B))\setminus I=\emptyset$, let $I_G=I\cup g(Z\cap A)$ and $I_H=I\cup g(Z\cap B)$,
    and let $\alpha(Z)\coloneqq \mu_G(I_G)+\mu_H(I_H)$.
    \item We define $\mu^*(I)$ as the minimum such $\alpha(Z)$ over all $Z\in \mathcal{Z}$ with $(g(Z\cap A)\cap g(Z\cap B))\setminus I=\emptyset$. Note that such a set $Z$ exists as $A$ is such a vertex cover.
\end{itemize}

\begin{claim}\label{claim:correctnessmu}
    The above procedure correctly computes $\mu_F$, that is, $\mu^*(I)=\mu_F(I)$ for every $I\subseteq [k]$.
\end{claim}
\begin{clproof}
    Let $I\subseteq [k]$. 
    
    We first verify that $\mu^*(I)\ge \mu_F(I)$.
    By definition, there exists  
    $Z\in \mathcal{Z}$ with $(g(Z\cap A)\cap g(Z\cap B))\setminus I=\emptyset$ such that \[\mu^*(I)=\alpha(Z)=\mu_G(I\cup g(Z\cap A))+\mu_H(I\cup g(Z\cap B)).\] 
    
    Let $T_G$ be a minimum vertex cover of $G$ with respect to $I\cup g(Z \cap A)$ and $T_H$ be a minimum vertex cover of $H$ with respect to $I\cup g(Z\cap B)$. Since $Z$ is a vertex cover of $Q$, $F-(T_G\cup T_H)$ has no edge between $V(G)$ and $V(H)$. Also, $T_G$ and $T_H$ are vertex covers of $G$ and $H$, respectively. As $(g(Z\cap A)\cap g(Z\cap B))\setminus I=\emptyset$,  $T_G\cup T_H$ is a vertex cover of $F$ with respect to $I$. So, we have  $\mu^*(I)\ge \mu_F(I)$.

    To show that $\mu^*(I)\le \mu_F(I)$, let $T$ be a minimum vertex cover of $F$ with respect to $I$, which has size $\mu_F(I)$. Let $T_G=T\cap V(G)$ and $T_H=T\cap V(H)$. Since $\full_F(T)=I$, for every $i\in [k]\setminus I$, either $i\notin \full_G(T_G)$ or $i\notin \full_H(T_H)$.
    Also, since $T$ is a vertex cover of $F$, for every $(i,j)\in M$, either $i\in \full_G(T_G)$ or $j\in \full_H(T_H)$.
    This implies that the set $\{a_i:i\in \full_G(T_G)\}\cup \{b_i:i\in \full_H(T_H)\}$ is a vertex cover of $Q$. 
    Note that $T_G$ is a minimum vertex cover of $G$ with respect to $\full_G(T_G)$; if there is a smaller vertex cover of $G$  with respect to $\full_G(T_G)$, then we can find a smaller vertex cover of $F$. Similarly, $T_H$ is also a minimum vertex cover of $H$ with respect to $\full_H(T_H)$. Then $\mu_F(I)=\alpha(Z)$ where $Z=\{a_i:i\in \full_G(T_G)\}\cup \{b_i:i\in \full_H(T_H)\}$.  Thus, $\mu^*(I)\le \mu_F(I)$.
\end{clproof}

Note that the number of vertex covers of $Q$ is at most $2^{2k}$. For each $I\subseteq [k]$, once we have computed $\alpha(Z)$ for all $Z\in\mathcal{Z}$, we can obtain $\mu(I)$ by taking the minimum of these values. Thus, each value $\mu(I)$, and hence the whole function $\mu^*$ can be computed in time $2^{\mathcal{O}(k)}$.

Now, we compute $\Pi(F)$ and a collection of minimum $\beta$-sets for $\beta\in \Pi(F)$. 
We construct sets $\Pi^*$ and $\mathcal{I}^*$ and will show that $\Pi^*=\Pi(F)$ and $\mathcal{I}^*$ is a collection of minimum $\beta$-sets for $\beta\in \Pi(F)$.
For a function $\beta:2^{[k]}\to \{0,1,2\}$, we need to determine whether there is a $\beta$-set in $F$.
Let $\beta:2^{[k]}\to \{0,1,2\}$ be a function. 
\begin{enumerate}
    \item Let $I\subseteq [k]$. We say that a pair $(I_G, I_H)$ of subsets of $[k]$ is a \emph{split of $I$ with respect to $M$} if \begin{itemize}
        \item $I_G\cap I_H=I$, and
        \item for every $(i,j)\in M$, $i\in I_G$ or $j\in I_H$.
    \end{itemize}
    A split $(I_G, I_H)$ of $I$ is \emph{proper} if $\mu_G(I_G)+\mu_H(I_H)=\mu_F(I)$. 
    \item A pair $(\beta_G, \beta_H)$ of functions $\beta_G, \beta_H: 2^{[k]}\to \{0,1,2\}$ is  \emph{legitimate for $\beta$} if for all  subsets $I\subseteq [k]$, we have 
    \[ \beta(I)=\min \left(2,  \sum_{\substack{(I_G, I_H) :\\ \text{proper split of $I$}}} \Big( \beta_G(I_G)\times \beta_H(I_H) \Big) \right).\]
    \item Assume there is a legitimate pair $(\beta_G, \beta_H)$ for $\beta$ where $\beta_G\in \Pi(G)$ and $\beta_H\in \Pi(H)$ and $S_G$ is a minimum $\beta_G$-set in $\mathcal{I}_G$ and $S_H$ is a minimum $\beta_H$-set in $\mathcal{I}_H$. Then we add $\beta$ to $\Pi^*$ and add $S_G\cup S_H$ to $\mathcal{I}^*$. Otherwise, we do not add.
\end{enumerate}

\begin{claim}\label{claim:merge}
    The above procedure correctly computes $\Pi(F)$, that is, $\Pi^*=\Pi(F)$. Also, $\mathcal{I}^*$ is a collection of minimum $\beta$-sets for $\beta\in \Pi(F)$.
\end{claim}
\begin{clproof}
    First assume that $\beta\in \Pi(F)$. Then there is a $\beta$-set $S$ in $F$. 
    Let $S_G=S\cap V(G)$ and $S_H=S\cap V(H)$, and 
    let $\beta_{G}$ be the characteristic of $S_G$ in $G$, and $\beta_{H}$ be the characteristic of $S_H$ in $H$. Clearly, $\beta_{G}\in \Pi(G)$ and $\beta_{H}\in \Pi(H)$.

    We claim that $(\beta_{G}, \beta_{H})$ is a legitimate pair. By the construction, this will imply that $\beta\in \Pi^*$. 
    
    Let $I\subseteq [k]$. For a minimum vertex cover $T$ of $F$ with respect to $I$ and containing $S$, let 
    \begin{itemize}
        \item $T_G=T\cap V(G)$ and $T_H=T\cap V(H)$,
        \item $I_G=\full_G(T_G)$ and $I_H=\full_H(T_H)$.
    \end{itemize}  Since $T$ is a vertex cover of $F$ with respect to $I$, $(I_G, I_H)$ is a split of $I$ with respect to $M$. Observe that $T_G$ is a minimum vertex cover of $G$ with respect to $I_G$ and $T_H$ is a minimum vertex cover of $H$ with respect to $I_H$, as $T$ is a minimum vertex cover of $F$ with respect to $I$. Thus, $(I_G, I_H)$ is proper. 

    Let $t=\beta(I)$. 
    If $t=0$, then there is no minimum vertex cover $T$ with respect to $I$ and containing~$S$. Then for every proper split $(I_G, I_H)$ of $I$ with respect to $M$, either $G$ has no minimum vertex cover with respect to $I_G$ containing $S_G$, or $H$ has no minimum vertex cover with respect to $I_H$ containing $S_H$.
    Thus,  we have
    \[\sum_{(I_G, I_H)\text{ : proper split of $I$}} \Big( \beta_{G}(I_G)\times \beta_{H}(I_H) \Big) =0.\]

    Assume $t=1$. Then there is a unique minimum vertex cover $T$ with respect to $I$ and containing~$S$. So, there is a unique proper split $(I_G, I_H)$ of $I$ with respect to $M$, for which $G$ has a unique minimum vertex cover with respect to $I_G$ containing $S_G$, and  $H$ has a unique minimum vertex cover with respect to $I_H$ containing $S_H$.
    Thus, we have 
    \[\sum_{(I_G, I_H)\text{ : proper split of $I$}} \Big( \beta_{G}(I_G)\times \beta_{H}(I_H) \Big) =1.\]

    Lastly, assume $t=2$. Then there are at least two minimum vertex covers $T$ with respect to $I$ and containing $S$. Then either 
    \begin{itemize}
        \item there are at least two proper splits $(I_G, I_H)$ of $I$ with respect to $M$ where 
        $G$ has a minimum vertex cover with respect to $I_G$ containing $S_G$, and  $H$ has a minimum vertex cover with respect to $I_H$ containing $S_H$,
        or
        \item there is a proper split $(I_G, I_H)$ of $I$ with respect to $M$ where 
        either 
        \begin{itemize}
            \item $G$ has at least two minimum vertex cover with respect to $I_G$ containing $S_G$ and $H$ has a minimum vertex cover with respect to $I_H$ containing $S_H$, or 
            \item $G$ has a minimum vertex cover with respect to $I_G$ containing $S_G$ and $H$ has at least two minimum vertex covers with respect to $I_H$ containing $S_H$.
        \end{itemize}
    \end{itemize} 
    Therefore, we have 
    \[\sum_{(I_G, I_H)\text{ : proper split of $I$}} \Big( \beta_{G}(I_G)\times \beta_{H}(I_H) \Big) \ge 2,\]
    as required. 

    \medskip
    Now, for the other direction, suppose that $\beta\in \Pi^*$. Then there is a legitimate pair $(\beta_G, \beta_H)$ where $\beta_G\in \Pi(G)$ and $\beta_H\in \Pi(H)$.
    So, there is a $\beta_G$-set $S_G$ of $G$ and a $\beta_H$-set $S_H$ of $H$. Let $S=S_G\cup S_H$. We claim that $S$ is a $\beta$-set. This will imply that $\beta\in \Pi(F)$.

    Let $I\subseteq [k]$ and $t=\beta(I)$. 
    Assume that $t=0$. Since $(\beta_G, \beta_H)$ is legitimate for $\beta$, 
     \[ 0=\beta(I)=\sum_{(I_G, I_H)\text{ : proper split of $I$}} \Big( \beta_G(I_G)\times \beta_H(I_H) \Big).\]
     Since $S_G$ is a $\beta_G$-set and $S_H$ is a $\beta_H$-set, the above equality implies that for every proper split $(I_G, I_H)$ of $I$, either there is no minimum vertex cover of $G$ with respect to $I_G$ and containing $S_G$, or there is no minimum vertex cover of $H$ with respect to $I_H$ and containing $S_H$. Thus, there is no minimum vertex cover of $F$ with respect to $I$ and containing $S$.

     Assume that $t=1$. Since $(\beta_G, \beta_H)$ is legitimate for $\beta$, 
     \[ 1=\beta(I)=\sum_{(I_G, I_H)\text{ : proper split of $I$}} \Big( \beta_G(I_G)\times \beta_H(I_H) \Big).\]
     This shows that there is a unique proper split $(I_G, I_H)$ of $I$ such that $G$ has a unique minimum vertex cover with respect to $I_G$ and containing $S_G$, and  $H$ has a unique minimum vertex cover with respect to $I_H$ and containing $S_H$. Thus, there is a unique minimum vertex cover of $F$ with respect to $I$ and containing $S$.

     Assume $t=2$.
     In this case, we have
     \[ \sum_{(I_G, I_H)\text{ : proper split of $I$}} \Big( \beta_G(I_G)\times \beta_H(I_H) \Big)\ge 2.\]
     Thus, either 
    \begin{itemize}
         \item there are at least two proper splits $(I_G, I_H)$ of $I$ with respect to $M$ where $G$ and $H$ have minimum vertex covers with respect to $I_G$ and $I_H$, respectively, or
        \item there is a proper split $(I_G, I_H)$ of $I$ with respect to $M$ where one of $G$ and $H$ has at least two minimum vertex covers with respect to $I_G$ or $I_H$, and the other has a minimum vertex cover with respect to $I_G$ or $I_H$. 
    \end{itemize} 
    By combining vertex covers of $G$ and $H$, there are at least two minimum vertex covers of $F$ with respect to $I$ and containing $S$.

    As discussed above, when $(\beta_G, \beta_H)$ is a legitimate pair for $\beta$ where $\beta_G\in \Pi(G)$ and $\Pi(H)$ and $S_G$ is a minimum $\beta_G$-set and $S_H$ is a minimum $\beta_H$-set, $S_G\cup S_H$ is a minimum $\beta$-set in $G$. Thus, $\mathcal{I}^*$ is a collection of $\beta$-sets for $\beta\in \Pi(F)$.
\end{clproof}

Observe that the number of possible functions $\beta:2^{[k]}\to \{0,1,2\}$ is at most $3^{2^{k}}$. Let $\beta$ be such a function. For each $I\subseteq [k]$, there are at most $3^{k-|I|}\le 3^k$ splits of $I$ with respect to $M$. Thus, for a fixed pair $(\beta_G, \beta_H)$ of functions, one can test whether $(\beta_G, \beta_H)$ is legitimate for $\beta$ in time $\mathcal{O}( 2^{k} \cdot 3^{k})$. Thus, we can determine $\Pi^*$ and $\mathcal{I}^*$ in time $2^{\mathcal{O}(2^k)}|V(F)|$.
This concludes the proof.
\end{proof}

\begin{lemma}\label{lem:labelchange}
    Let $k$ be a positive integer, and let $(G, \lab_G)$ be a $k$-labeled graph. Let $R:[k]\to [k]$ be a function and let $(F, \lab_F)=\rho_R(G, \lab_G)$. 
    
    Given $\Pi(G)$, $\mu_G$, and a collection $\mathcal{I}_G$ of minimum $\beta$-sets for $\beta\in \Pi(G)$, one can compute $\Pi(F)$, $\mu_F$ and a collection $\mathcal{I}_F$ of minimum $\beta$-sets for $\beta\in \Pi(F)$ in time $2^{\mathcal{O}(2^k)} |V(G)|$.
\end{lemma}
\begin{proof}
We first compute $\mu_F$. Observe that for each $I\subseteq [k]$,
$\{v\in V(F):\lab_F(v)\in I\}$ is the same as $\{v\in V(G):\lab_G(v)\in R^{-1}(I)\}$.
Thus, we have $\mu_F(I)=\mu_G(R^{-1}(I))$. So, $\mu_F$ can be computed in time $\mathcal{O}(2^k)$.

Next, we compute $\Pi(F)$ and $\mathcal{I}_F$. 
We construct sets $\Pi^*$ and $\mathcal{I}^*$. Given a function $\beta:2^{[k]}\to \{0,1,2\}$, we define a function $\widehat{\beta}$ whose domain is $\mathcal{U}=\{R^{-1}(I):I\subseteq [k]\}$ such that for every $J=R^{-1}(I)$ with $I\subseteq [k]$, we have $\widehat{\beta}(J)=\beta(I)$.
We do the following.
\begin{itemize}
    \item Assume that there is a function $\gamma:2^{[k]}\to \{0,1,2\}$ such that $\gamma_{|_{\mathcal{U}}}=\widehat{\beta}$.
    Then we add $\beta$ to $\Pi^*$. 
    Among all functions $\gamma:2^{[k]}\to \{0,1,2\}$ such that $\gamma_{|_{\mathcal{U}}}=\widehat{\beta}$, 
    we choose one where the $\gamma$-set $S$ in $\mathcal{I}_G$ has minimum number of vertices, and we add this $S$ to $\mathcal{I}^*$ as a minimum $\beta$-set.
    \item If there is no such a function $\gamma$, then we do not add for $\beta$. 
    \end{itemize}  
    It is straightforward to verify that $\Pi^*=\Pi(F)$ and $\mathcal{I}^*$ is a collection of minimum $\beta$-sets for $\beta\in \Pi(F)$.    This can be done in time $2^{\mathcal{O}({2^k})} |V(G)|$.
\end{proof}

Now, we are ready to prove Theorem~\ref{thm:maincliquewidth}.

\begin{proof}[Proof of Theorem~\ref{thm:maincliquewidth}]
    Using an algorithm by Oum~\cite{Oum2009}, we can compute a clique-width $(2^{3t+2}-1)$-expression of a graph of clique-width $t$ in time $\mathcal{O}(8^t |V(G)|^4)$ as explained in the preliminary section. In the rest, we discuss how to obtain an algorithm if a clique-width expression is given. 

    Let $G$ be a graph and assume that its clique-width $k$-expression is given. 
    By Theorem~\ref{thm:transform}, we can transform it into an NLC-width $k$-expression $\phi$  in polynomial time.

    We design a bottom-up dynamic programming along the NLC-width $k$-expression. At each $k$-labeled graph $(F, \lab_F)$ arising in $\phi$, we compute sets $\Pi(F)$, $\mu_F$, and a collection $\mathcal{I}_F$ of minimum $\beta$-sets for $\beta\in \Pi(F)$ as follows.
    \begin{enumerate}
    \item Assume $(F, \lab_F)=i(x)$, that is, $F$ is a graph on a vertex $x$ with label $i$.  
    \begin{itemize}
        \item Observe that $\mu_F(\{i\})=1$ because $\{x\}$ is the unique minimum vertex cover of $F$ with respect to $\{x\}$.  Also, $\mu_F(\emptyset)=0$ because $\emptyset$ is the unique minimum vertex cover of $F$ with respect to $\emptyset.$ For other subsets $I$ of $[k]$, $\mu_F(I)=\infty$, as there is no vertex cover of $F$ with respect to $I$.
        \item Note that the empty set has the characteristic $\beta_0$ where $\beta_0(J)=1$ for $J=\{i\}$ or $\emptyset$, and $\beta_0(J)=0$ otherwise. The set $\{x\}$ has characteristic $\beta_1$ where $\beta_1(J)=1$ for $J=\{i\}$, and $\beta_1(J)=0$ otherwise. 
        Let $\Pi(F)=\{\beta_0, \beta_1\}$.
        We store the empty set as a minimum $\beta_0$-set, and $\{x\}$ as a minimum $\beta_1$-set. 
        \item These can be computed in time $\mathcal{O}(k)$.
    \end{itemize}
    \item Assume that $(F, \lab_F)=\rho_R(F_1, \lab_1)$ for some function $R:[k]\to [k]$. By Lemma~\ref{lem:labelchange}, we can in time $2^{\mathcal{O}(2^k)} |V(F)|$ compute  $\Pi(F)$, $\mu_F$, and a collection $\mathcal{I}_F$ of minimum $\beta$-sets for $\beta\in \Pi(F)$.
    \item Assume that $(F, \lab_F)=(F_1, \lab_1)\times_M (F_2, \lab_2)$ for some $M\subseteq [k]^2$. By Lemma~\ref{lem:product}, we can in time $2^{\mathcal{O}(2^k)} |V(F)|$ compute $\Pi(F)$, $\mu_F$, and a collection $\mathcal{I}_F$ of minimum $\beta$-sets for $\beta\in \Pi(F)$.
\end{enumerate}
At the end, by Lemma~\ref{lem:inclusionfinalstep}, we can solve \textsc{Min PAU-VC} in time $2^{\mathcal{O}(2^k)} |V(G)|$. Note that there are at most $\mathcal{O}(k^2 |V(G)|)$ operations in the NLC-width $k$-expression.  Thus, in total, we can solve \textsc{Min PAU-VC} in time $2^{\mathcal{O}(2^k)} |V(G)|^2$.
\end{proof}

\section{Unit interval graphs}\label{sec:unitinterval}
In this section, we give a linear time algorithm for \textsc{Min PAU-VC} on unit interval graphs.
A graph is a unit interval graph if there is a set $\mathcal I$ of intervals of length one on the real line so that $G$ is the intersection graph of $\mathcal I$. We refer to Figure~\ref{fig:unit_interval}.

\begin{figure}[t]
\includegraphics[width=\columnwidth]{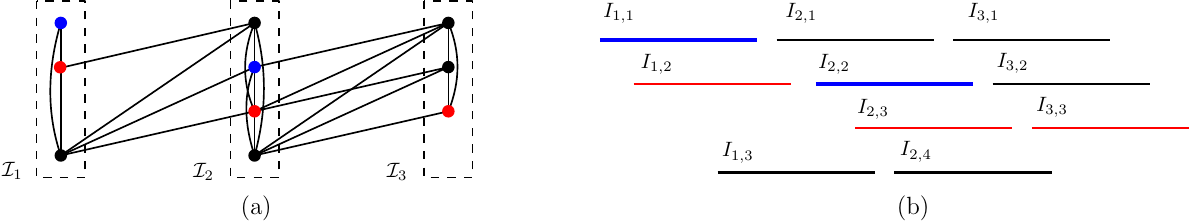}
\caption{An illustration of the algorithm for a unit interval graph $G$. Figures (a) and (b) represent the unit interval graph $G$ and its representation, respectively.
The algorithm returns $\{I_{1,1}, I_{2,2}\}$ as an optimal solution of \textsc{Min PAU-VC} of $G$, where $\{I_{1,1}, I_{2,2}\}= S[3,3]$ as $S[1,2]=\{I_{1,1}\}, A_{1,2,3}=\{I_{2,2}\},$ and $A_{2,3,3}=\emptyset$. Precisely, $S[3,3]=S[1,2]\cup A_{1,2,3}\cup A_{2,3,3}$. }
\label{fig:unit_interval}
\end{figure}
\begin{theorem}\label{thm:unitinterval}
    \textsc{Min PAU-VC} can be solved in linear time on unit interval graphs.
\end{theorem}
\begin{proof}
Let $G$ be a given unit interval graph. 
We can find a set $\mathcal I$ of unit intervals representing $G$ in linear time~\cite{corneil1995simple}.
Furthermore, the obtained $\mathcal I$ is sorted by the left end points.
For clarity, we refer to the intervals in $\mathcal I$ as the vertices of~$G$. By perturbing if necessary, we may assume that all intervals are pairwise distinct.

First, we find a maximum independent set $\{\tilde I_1,\dots, \tilde I_m\}$ of $G$ as follows: $\tilde I_1$ is the leftmost interval in $\mathcal I$ and $\tilde I_{i+1}$ is the leftmost interval in $\mathcal I$ disjoint to $\tilde I_j$ for all $j\in[i]$. 
It is easy to see that the obtained set is a maximum independent set of $G$.
Then, for each $\tilde{I}_i$, let $\mathcal{I}_i$ be the set of intervals in $\mathcal{I}$ that start after $\tilde{I}_i$ and intersect with $\tilde{I}_i$, along with $\tilde{I}_i$ itself.
Then $\{\mathcal I_1,\dots, \mathcal I_m\}$ gives a partition of $\mathcal I$, and each $\mathcal I_i$ forms a clique in $G$.
Note that if two intervals $I\in \mathcal I_i$ and $J\in \mathcal I_j$ are intersecting, then $|i-j|\leq 1$ since every interval has a unit length.

\begin{claim}\label{lem:max_excluding}
    Every minimum vertex cover of $G$ excludes exactly one vertex in $\mathcal I_i$ for each $i\in[m]$. 
\end{claim}
\begin{clproof}
    It is well known that the complement of a maximum independent set is a minimum vertex cover, and vice versa.
    Since $m$ is the size of the maximum independent set of $G$, every minimum vertex cover should exclude $m$ vertices.
    Note that each $\mathcal I_i$ forms a clique in $G$, it cannot exclude more than one vertex from the same $\mathcal I_i$.
    Thus, exactly one vertex for each $\mathcal I_i$ is excluded.
\end{clproof}

Now we describe a dynamic programming to solve \textsc{Min PAU-VC} for $G$.

Let $i\in [m]$ and $j\in [|\mathcal I_i|]$.
Let $I_{i,j}$ be the $j$th leftmost interval in $\mathcal I_i$, and let $G_{i,j}$ be the subgraph of $G$ induced by the intervals that start before $I_{i,j}$ together with $I_{i,j}$.

For $j\in [|\mathcal I_1|]$, let 
$A_{1,j}\coloneqq \{I_{1,i}:i<j\}$. 
For $1\le i\le m-1$ and $a\in [|\mathcal I_i|]$ and $b\in [|\mathcal I_{i+1}|]$, let 
\begin{align*}
    A_{i,a,b}\coloneqq
    &\{I_{i,x} : a<x \textnormal{ and } I_{i,x}\cap I_{i+1,b}=\emptyset\} \\
    &\cup \{I_{i+1,y} : y<b \textnormal{ and } I_{i+1,y}\cap I_{i,a}=\emptyset\}.
\end{align*}
Intuitively, $A_{1,j}$ and $A_{i,a,b}$ are sets of vertices that have to be included in a set forcing a unique minimum vertex cover $S$ when we want that $I_{i,a}$ and $I_{i+1, b}$ are not in $S$.

Now, for each $I_{i,j}\in \mathcal I$, we denote by $S[i,j]$ a smallest vertex set in $G_{i,j}$ such that among all minimum vertex covers of $G_{i,j}$, exactly one contains $T$, and the unique vertex cover excludes $I_{i,j}$. 
 
We compute $S[i,j]$ in a lexicographic order.
As the base case, we set $S[1,j]\coloneqq A_{1,j}$ for each $j\in [|\mathcal I_1|]$.
By assuming that every $S[i,j']$ has been computed for $j'\in [|\mathcal I_i|]$, we set $S[i+1,j]$ as the smallest vertex set among
\begin{align*}
    S[i,j'] \cup A_{i,j',j}
\end{align*}
where $I_{i,j'}\in \mathcal I_i$ is disjoint from $I_{i+1,j}$.
Note that $S[i+1, j]$ is well-defined, because $I_{i,1}=\tilde I_i$ is disjoint from $I_{i+1,1}=\tilde I_{i+1}$ along with $I_{i+1,j}$.

\begin{claim}\label{claim:correct_unit_int}
    For every $i\in [m]$ and $j\in [|\mathcal I_i|]$, $S[i,j]$ is a smallest vertex set in $G_{i,j}$ such that 
    among all minimum vertex covers of $G_{i,j}$, exactly one contains $S[i,j]$, and the unique vertex cover excludes~$I_{i,j}$. 
\end{claim}
\begin{clproof}
    We prove this by induction on $i+j$.
    First, assume that $i=1$. Then $G_{1,j}$ is a complete graph. Thus, there is only one vertex cover excluding $I_{1,j}$, namely $V(G_{1,j})\setminus \{I_{1,j}\}$. If $S[1,i]$ does not contain all vertices of $V(G_{1,j})\setminus \{I_{1,j}\}$, then we may include $I_{1,j}$ to obtain another minimum vertex cover of $G_{1,j}$. Thus, $S[1,j]$ should be exactly $V(G_{1,j})\setminus \{I_{1,j}\}$.

    Suppose that $i\ge 2$.
    Let $T$ be a minimum vertex set in $G_{i,j}$ such that 
     among all minimum vertex covers of $G_{i,j}$, exactly one contains $T$, and the unique vertex cover excludes~$I_{i,j}$. 
    
    We will show that $T$ is of the form $S[i-1,j']\cup A_{i-1,j',j}$ for some $j'\in [|\mathcal I_{i-1}|]$. Let $U$ be the unique minimum vertex cover in $G_{i,j}$ containing $T$, and let $W=V(G_{i,j})\setminus U$. 
    Note that $U$ does not contain $I_{i,j}$ and $W$ contains $I_{i,j}$.
    Observe that $W$ contains exactly one vertex from each of {$\mathcal I_1, \ldots, \mathcal I_i$}. Let $j_1, \ldots, j_{i-1}$ be integers such that $W=\{I_{1, j_1}, I_{2, j_2}, \ldots, I_{i-1, j_{i-1}}, I_{i,j}\}$.

    We claim that 
    \[A^*\coloneqq A_{1,j_1}\cup \left( \bigcup_{x\in [i-1]} A_{x,j_x,j_{x+1}}\right) \subseteq T. \]
    Suppose for contradiction that this is not true. 
    
    First assume that $A_{1,j_1}$ contains a vertex $I$ that is not in $T$.
    We show that $(U\setminus \{I\})\cup \{I_{1,j_1}\}$ is also a minimum vertex cover of $G_{i,j}$ containing $T$.  Note that $I_{2,j_2}$ is the unique vertex of $\mathcal{I}_2\cap V(G_{i,j})$ that is not contained in $U$. Since $I$ and $I_{2,j_2}$ are not adjacent, all neighbors of $I$ in $\mathcal{I}_2$ are contained in $U$. Thus, all edges between $I$ and $\mathcal{I}_2$ are covered by the neighbors of $I$. Moreover, the edges between $I$ and $\mathcal{I}_1$ are also covered by the neighbors of $I$. Therefore, replacing $I$ with $I_{1,j_1}$ yields a vertex cover.
    This implies that there are two minimum vertex covers of $G_{i,j}$ containing $T$, a contradiction.

    We assume that for some $x\in [i-1]$, $A_{x,j_x,j_{x+1}}$ contains a vertex $I$ that is not in $T$. If 
    \[I\in \{I_{x,z} : j_x<z \textnormal{ and } I_{x,z}\cap I_{x+1,j_{x+1}}=\emptyset\},\]
    then $(U\setminus \{I\})\cup \{I_{x, j_x}\}$ is a minimum vertex cover of $G_{i,j}$. Indeed, if this is not a vertex cover, then there is an edge $IJ$ that is not covered by $(U\setminus \{I\})\cup \{I_{x, j_x}\}$, and we have 
    \[J\in \{I_{x+1,y}:y<b\text{  and }I_{x+1,y}\cap I_{x,j_x}=\emptyset\}.\]
    But then $J$ is a vertex in $U$, and $IJ$ is covered. This means that there are two minimum vertex covers containing $T$, a contradiction. 
    Otherwise, if 
    \[I\in \{I_{x+1,z} : z<j_{x+1} \textnormal{ and } I_{x+1,z}\cap I_{x,j_x}=\emptyset\},\]
    then by a similar argument, $(U\setminus \{I\})\cup \{I_{x+1, j_{x+1}}\}$ is a minimum vertex cover of $G_{i,j}$, a contradiction.  Therefore, $A^*\subseteq T$. 

    Now, we verify that $A^*$ already forces that $U$ is a unique minimum vertex cover containing $A^*$. Suppose there is another minimum vertex cover $U'$ containing $A^*$. As $U'\neq U$, $U'$ does not contain a vertex of $U\setminus A^*$. Then $G_{i,j}-U'$ cannot contain an independent set of size $i$, a contradiction. 
    
    By our construction, $T$ is $S[i-1,j']\cup A_{i-1,j',j}$ for some $j'$ where $I_{i-1,j'}\in \mathcal{I}_{i-1}$ disjoint from $I_{i, j}$. On the other hand, we compute $S[i,j]$ as a smallest vertex set among all possible $S[i-1,j']\cup A_{i-1, j', j}$. Thus, $S[i,j]$ is a smallest vertex set in $G_{i,j}$ that forces a unique minimum vertex cover in $G_{i,j}$ and the vertex cover excludes $I_{i,j}$.
\end{clproof}

Furthermore, the smallest vertex set among \[S[m,j]\cup \{I_{m,k}\in \mathcal I_m : j<k\}\] for $j\in[|\mathcal I_m|]$ is an optimal solution of \textsc{Min PAU-VC} for $G$.

This algorithm returns a solution in polynomial time.
In the following, we slightly modify it as a linear time algorithm.

\subparagraph{Linear time algorithm.}
We compute the interval set $\mathcal I$ corresponding to the given unit interval graph $G$, together with its decomposition $\mathcal I_1,\dots, \mathcal I_m$ as described above. It takes a linear time.
To obtain a linear time algorithm for \textsc{Min PAU-VC}, 
we compute and store the size $s[i+1,j]$ of the set $S[i+1,j]$ for each $(i+1,j)$ with $i\in [m-1]$ and $j\in [|\mathcal I_{i+1}|]$ instead of explicitly constructing the set $S[i+1,j]$.
Since the set $S[i+1,j]$ is the union of disjoint sets $S[i,j']$ and $A_{i,j',j}$, we can compute $s[i+1,j]$ without explicitly constructing $S[i,j]$.
Furthermore, we also store the index $j'\in [|\mathcal I_i|]$ at the pair $(i+1,j)$ so that $S[i+1,j]=S[i,j']\cup A_{i,j',j}$.
We can compute all values of $s[\cdot,\cdot]$ in $\mathcal{O}(|\mathcal I|)$ time.

\begin{claim}\label{claim:bound_opearation}
    We can compute all $s[\cdot,\cdot]$ in $\mathcal{O}(|\mathcal I|)$ time.
\end{claim}
\begin{clproof}
    We set $s[1,j]=j-1$ by the definition.
    For an index $i\in [m-1]$, we assume that every $s[i,\cdot]$ is computed already, and describe how to compute all $s[i+1,\cdot]$ in $\mathcal O(|\mathcal I_i\cup \mathcal I_{i+1}|)$ time which directly implies the claim.
    For this, we first compute an index $k(j)\in [|\mathcal I_{i}|]$ for each $j\in [|\mathcal I_{i+1}|]$ so that 
    the interval $I_{i,k(j)}$ is the rightmost interval in $\mathcal I_i$ disjoint from $I_{i+1,j}$.
    Since $\mathcal I_i$ and $\mathcal I_{i+1}$ are sorted, the indices $k(j)$'s are monotonically increasing. Furthermore, we can compute all $k(j)$'s in $\mathcal O(|\mathcal I_i\cup \mathcal I_{i+1}|)$ time.
    For clarity, we set $k(0)=0$ in the following.
    Note that $I_{i,j'}\cap I_{i+1,j}=\emptyset$ if and only if $j'\leq k(j)$ for $j\in [|\mathcal I_{i+1}|]$.
    
    Recall that $s[i+1,j]$ is the smallest value among $s[i,j']+|A_{i,j',j}|$ with $j'\leq k(j)$.
    Furthermore, if $j'\leq k(j-1)$, then $A_{i,j',j}$ is the same as 
    \begin{align*}
        A_{i,j',j-1} \cup \{I_{i,x} : k(j-1)<x\leq k(j)\}\cup \{I_{i+1,j-1}\}.
    \end{align*}
    If an index $j'\leq k(j)$ gives the smallest set $S[i,j']\cup A_{i,j',j}$, then either $j'>k(j-1)$ or it gives a smallest set among $S[i,j']\cup A_{i,j',j-1}$.
    Thus, we can compute the size $s[i+1,j]$ of $S[i+1,j]$ by comparing $k(j)-k(j-1)+1$ values.
    Totally, computing all $s[i+1,\cdot]$ requires $\mathcal O(|\mathcal I_i\cup \mathcal I_{i+1}|)$ time, and thus, computing all $s[\cdot,\cdot]$ takes $\mathcal O(|\mathcal I|)$ time. 
\end{clproof}

After we compute every $s[\cdot,\cdot]$, we find out the index $j^*\in [|\mathcal I_m|]$ minimizing the value $s[m,j^*]+|\mathcal I_m|-j^*$.
Then we define $j_m=j^*$ and $j_{i}$ as the index with $S[i+1,j_{i+1}]=S[i,j_i]\cup A_{i,j_i,j_{i+1}}$ for $i\in [m-1]$.
By the definition of the sets $S[\cdot,\cdot]$, the following set is the same as $S[m,j^*]\cup \{I_{m,k}\in \mathcal I_m : j_m<k\}$ that is an optimal solution of \textsc{Min PAU-VC}
\[A_{1,j_1}\cup \left( \bigcup_{i\in [m-1]} A_{i,j_i,j_{i+1}}\right) \cup \{I_{m,k}\in \mathcal I_m : j_m<k\}.\]

In conclusion, our algorithm returns a solution of \textsc{Min PAU-VC} for $G$ in linear time, and thus, Theorem~\ref{thm:unitinterval} holds.
\end{proof}
\section{Split graphs}\label{sec:split}
In this section, we describe a linear time algorithm for split graphs.
A split graph $G$ is a graph in which there exist disjoint subsets $A,B\subseteq V(G)$ such that $V(G)=A\cup B$, $A$ is a clique and $B$ is an independent set.
\begin{theorem}
    \textsc{Min PAU-VC} can be solved in linear time on split graphs.
\end{theorem}
\begin{proof}
Let $G$ be a split graph with a partition $(A,B)$ of its vertex set such that $A$ is a clique and $B$ is an independent set.
Observe that a minimum vertex cover excludes at most one vertex from $A$, and furthermore, $A$ is a vertex cover of $G$.
We claim that it is safe to remove every vertex in $A$ that has at least two neighbors in $B$, and also to remove all isolated vertices. This can be done in linear time. 
\begin{claim}\label{claim:splittwoneighbor}
    If $v\in A$ has at least two adjacent vertices in $B$, then every minimum vertex cover of $G$ includes $v$.
\end{claim}
\begin{clproof}
    For a contradiction, suppose that there is a minimum vertex cover $X$ in $G$ excluding such a vertex $v$. We show that $|X|$ is strictly larger than $|A|$.
    This directly implies the claim.
    As $A$ is a clique in $G$,
    $X$ excludes only $v$ in $A$.
    Furthermore, $X$ includes all vertices in $B$ adjacent to $v$, where there are at least two such vertices by the assumption. 
    Thus, $|X|\ge |A|+1$.
\end{clproof} 
\begin{figure}[t]
\centering
\includegraphics[width=0.5\columnwidth]{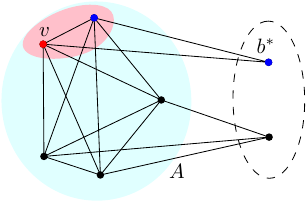} %
\caption{$(N_G(b^*)\setminus \{v\})\cup \{b^*\}$ is the minimum vertex set so that exactly one minimum vertex cover $(A\setminus \{v\})\cup \{b^*\}$ in the split graph includes it.}
\label{fig:split}
\end{figure}

In the following, suppose that every vertex in $A$ has at most one adjacent vertex in $B$ and there is no isolated vertex.
Let $A_0$ be the set of vertices in $A$ which has no adjacent vertex in $B$. 
We first consider the case that $A_0$ is not empty.
In such a case, a minimum vertex cover of $G$ has size $|A|-1$. Furthermore, 
any minimum vertex cover is of the form $A \setminus {v}$ for some $v \in A_0$.
Thus, $A_0\setminus \{v\}$ is an optimal solution of \textsc{Min PAU-VC} for an arbitrary vertex $v\in A_0$.

In the following, we consider the other case that $A_0=\emptyset$.
In this case, the size of a minimum vertex cover of $G$ is~$|A|$. See Figure~\ref{fig:split}. 

For each $a\in A$, let $v_a$ be the vertex of $B$ that is adjacent to $a$.
Observe that for each $a\in A$, $(A\setminus \{a\})\cup \{v_a\}$ is also a minimum vertex cover.
We find a vertex $b^*\in B$ minimizing $N_G(b^*)$, and return $(N_G(b^*)\setminus \{v\})\cup \{b^*\}$ as a solution of \textsc{Min PAU-VC}, where $v$ is an arbitrary vertex in $N_G(b^*)$. 

\begin{claim}\label{claim:split}
    Let $b^*\in B$ such that $|N_G(b^*)|$ is minimum, and let $v\in N_G(B^*)$. Then
    $(N_G(b^*)\setminus \{v\})\cup \{b^*\}$ is a solution of \textsc{Min PAU-VC}.
\end{claim}
\begin{clproof}
    Figure~\ref{fig:split} illustrates this proof.
    First, we show that $S^*=(N_G(b^*)\setminus \{v\})\cup \{b^*\}$ forces a unique minimum vertex cover in $G$.
    Let $X$ be a minimum vertex cover of $G$ containing $S^*$.
    
    As $b^*\in X$, there is a vertex $a$ in $A$ excluded by $X$ due to $|X|=|A|$. If $a$ is not in $N_G(b^*)$, then the edge $av_a$ is not covered by $X$, a contradiction. Thus, $a=v$ and 
    $X=(A\setminus \{v\})\cup \{b^*\}$ is the unique minimum vertex cover in $G$ including $S^*$.

    Let $T$ be a minimum vertex set of $G$ forcing a unique minimum vertex cover in $G$. 
    We show that $|T|\ge |S^*|$.
    Note that there is no minimum vertex cover of $G$ containing two vertices in $B$.
    Thus, $T$ includes at most one vertex in $B$.
    If $T$ does not contain any vertex in $B$ and does not contain a vertex $u$ in $A$, then there are two minimum vertex covers $A$ and $(A\setminus \{u\})\cup \{v_u\}$, a contradiction.
    Thus, if $T$ does not contain any vertex in $B$, then $T=A$, and therefore, $|T|\ge |S^*|$.

    Suppose that $T$ contains exactly one vertex $b$ in $B$.
    In such a case, if $T$ excludes two distinct vertices $u$ and $u'$ of $N_G(b)$, then there are two minimum vertex covers $(A\setminus \{u\})\cup \{b\}$ and $(A\setminus \{u'\})\cup \{b\}$ in $G$ including $T$.
    Therefore, $T$ excludes at most one vertex in $N_G(b)$.
    As we chose $b^*$ in $B$ with minimum $|N_G(b^*)|$, we have 
    \[|T|\ge (|N_G(b)|-1)+1\ge |S^*|,\]
    as required.
\end{clproof}

In conclusion, we can find a solution of \textsc{Min PAU-VC} for a split graph $G$ by checking the number of neighbors for each vertex in $B$. Thus, it takes $\mathcal{O}(|V(G)|)$ time, because every vertex in $A$ has at most one neighbor in $B$.
\end{proof}

\section{Conclusion}\label{sec:conclusion}
In this paper, our main contributions are three-fold: a fixed-parameter tractable algorithm for \textsc{Min PAU-VC} parameterized by clique-width,
and linear-time algorithms for unit interval graphs and split graphs. 
In particular, the first algorithm improves the best-known algorithm for \textsc{PAU-VC} on trees significantly. 
We believe that these algorithms can be used to generate benchmark datasets for evaluating the performances of AI algorithms on the unique vertex cover problem. 

There are still lots of open problems in this topic. Can we design polynomial-time algorithms for interval graphs, chordal graphs, or perfect graphs? It is known that these graph classes admit polynomial-time algorithms for the minimum vertex cover problem~\cite{grotschel1981ellipsoid}. Can we reduce the dependency on clique-width to be single-exponential, or can we show that our algorithm is optimal? Recall that  our algorithm runs in time double exponential in the clique-width of a graph. Although the running time seems large, it is still possible that our algorithms are optimal; there are several problems with lower bounds that are double exponential in the tree-width or clique-width~\cite{marx2016double,golovach2018cliquewidth,FoucaudKLMS2024,BliznetsH2024}.

Lastly, one may ask whether approximation algorithms can be designed for $\textsc{Min PAU-VC}$.
For bipartite graphs, we can show that there is no polynomial-time algorithm even for constant-factor approximation, unless $\mathrm{P}=\mathrm{NP}$. This follows by combining the reduction of Horiyama et al.~\cite{horiyama2024theoretical} with the hardness result of Irving~\cite{IRVING1991197}. 
Horiyama et al.\ proved that \textsc{PAU-VC} is NP-complete on bipartite graphs by a reduction from \textsc{Minimum Independent Dominating Set} on bipartite graphs while preserving the minimum value.
Since the latter problem admits no polynomial-time constant approximation on bipartite graphs unless $\mathrm{P}=\mathrm{NP}$~\cite{IRVING1991197}, it also rules out any polynomial-time constant-factor approximation for \textsc{Min PAU-VC} on bipartite graphs.


\begin{thebibliography}{10}

\bibitem{AnCCKLOS2025}
Shinwoo An, Yeonsu Chang, Kyungjin Cho, O{-}joung Kwon, Myounghwan Lee, Eunjin
  Oh, and Hyeonjun Shin.
\newblock Pre-assignment problem for unique minimum vertex cover on bounded
  clique-width graphs.
\newblock {\em Proceedings of the AAAI Conference on Artificial Intelligence},
  39(25):26886--26894, Apr. 2025.

\bibitem{asuncion2007uci}
Arthur Asuncion, David Newman, et~al.
\newblock {UCI} machine learning repository, 2007.

\bibitem{BergougnouxKK2020}
Benjamin Bergougnoux, Mamadou~Moustapha Kant\'e, and O{-}joung Kwon.
\newblock An optimal {XP} algorithm for {H}amiltonian cycle on graphs of
  bounded clique-width.
\newblock {\em Algorithmica}, 82(6):1654--1674, 2020.

\bibitem{bertossi1984dominating}
Alan~A. Bertossi.
\newblock Dominating sets for split and bipartite graphs.
\newblock {\em Information processing letters}, 19(1):37--40, 1984.

\bibitem{BliznetsH2024}
Ivan Bliznets and Markus Hecher.
\newblock Tight double exponential lower bounds.
\newblock In {\em Theory and applications of models of computation}, volume
  14637 of {\em Lecture Notes in Comput. Sci.}, pages 124--136. Springer,
  Singapore, 2024.

\bibitem{bozeman2019restricted}
Chassidy Bozeman, Boris Brimkov, Craig Erickson, Daniela Ferrero, Mary Flagg,
  and Leslie Hogben.
\newblock Restricted power domination and zero forcing problems.
\newblock {\em Journal of Combinatorial Optimization}, 37:935--956, 2019.

\bibitem{calabro2008complexity}
Chris Calabro, Russell Impagliazzo, Valentine Kabanets, and Ramamohan Paturi.
\newblock The complexity of unique $k$-{SAT}: An isolation lemma for
  $k$-{CNF}s.
\newblock {\em Journal of Computer and System Sciences}, 74(3):386--393, 2008.

\bibitem{MR1480799}
Gary Chartrand, Heather Gavlas, Robert~C. Vandell, and Frank Harary.
\newblock The forcing domination number of a graph.
\newblock {\em J. Combin. Math. Combin. Comput.}, 25:161--174, 1997.

\bibitem{corneil1995simple}
Derek~G. Corneil, Hiryoung Kim, Sridhar Natarajan, Stephan Olariu, and Alan~P
  Sprague.
\newblock Simple linear time recognition of unit interval graphs.
\newblock {\em Information processing letters}, 55(2):99--104, 1995.

\bibitem{CorneilR2005}
Derek~G. Corneil and Udi Rotics.
\newblock On the relationship between clique-width and treewidth.
\newblock {\em SIAM J. Comput.}, 34(4):825--847, 2005.

\bibitem{Courcelle1990}
Bruno Courcelle.
\newblock The monadic second-order logic of graphs. {I}. {R}ecognizable sets of
  finite graphs.
\newblock {\em Inform. and Comput.}, 85(1):12--75, 1990.

\bibitem{CourcelleMR2000}
Bruno Courcelle, Johann~A. Makowsky, and Udi Rotics.
\newblock Linear time solvable optimization problems on graphs of bounded
  clique-width.
\newblock {\em Theory Comput. Syst.}, 33(2):125--150, 2000.

\bibitem{CourcelleO2000}
Bruno Courcelle and Stephan Olariu.
\newblock Upper bounds to the clique width of graphs.
\newblock {\em Discrete Appl. Math.}, 101(1-3):77--114, 2000.

\bibitem{demaine2018fewest}
Erik~D. Demaine, Fermi Ma, Ariel Schvartzman, Erik Waingarten, and Scott
  Aaronson.
\newblock The fewest clues problem.
\newblock {\em Theoretical Computer Science}, 748:28--39, 2018.

\bibitem{FellowsRRS2006}
Michael~R. Fellows, Frances~A. Rosamond, Udi Rotics, and Stefan Szeider.
\newblock Clique-width minimization is {NP}-hard (extended abstract).
\newblock In {\em S{TOC}'06: {P}roceedings of the 38th {A}nnual {ACM}
  {S}ymposium on {T}heory of {C}omputing}, pages 354--362. ACM, New York, 2006.

\bibitem{ferrero2018relationship}
Daniela Ferrero, Leslie Hogben, Franklin~H.J. Kenter, and Michael Young.
\newblock The relationship between $k$-forcing and $k$-power domination.
\newblock {\em Discrete Mathematics}, 341(6):1789--1797, 2018.

\bibitem{Fomin2010}
Fedor~V. Fomin, Petr~A. Golovach, Daniel Lokshtanov, and Saket Saurabh.
\newblock Intractability of clique-width parameterizations.
\newblock {\em SIAM J. Comput.}, 39(5):1941--1956, 2010.

\bibitem{Fomin2014}
Fedor~V. Fomin, Petr~A. Golovach, Daniel Lokshtanov, and Saket Saurabh.
\newblock Almost optimal lower bounds for problems parameterized by
  clique-width.
\newblock {\em SIAM J. Comput.}, 43(5):1541--1563, 2014.

\bibitem{Fomin2019}
Fedor~V. Fomin, Petr~A. Golovach, Daniel Lokshtanov, Saket Saurabh, and Meirav
  Zehavi.
\newblock Clique-width {III}: {H}amiltonian cycle and the odd case of graph
  coloring.
\newblock {\em ACM Trans. Algorithms}, 15(1):Art. 9, 27, 2019.

\bibitem{FominK2024}
Fedor~V. Fomin and Tuukka Korhonen.
\newblock Fast {FPT}-approximation of branchwidth.
\newblock {\em SIAM J. Comput.}, 53(4):1085--1131, 2024.

\bibitem{FoucaudKLMS2024}
Florent Foucaud, Esther Galby, Liana Khazaliya, Shaohua Li, Fionn Mc~Inerney,
  Roohani Sharma, and Prafullkumar Tale.
\newblock Problems in {NP} can admit double-exponential lower bounds when
  parameterized by treewidth or vertex cover.
\newblock In {\em 51st {I}nternational {C}olloquium on {A}utomata, {L}anguages,
  and {P}rogramming}, volume 297 of {\em LIPIcs. Leibniz Int. Proc. Inform.},
  pages Art. No. 66, 19. Schloss Dagstuhl. Leibniz-Zent. Inform., Wadern, 2024.

\bibitem{gabow1999unique}
Harold~N. Gabow, Haim Kaplan, and Robert~E. Tarjan.
\newblock Unique maximum matching algorithms.
\newblock In {\em Proceedings of the thirty-first annual ACM symposium on
  Theory of Computing}, pages 70--78, 1999.

\bibitem{golovach2018cliquewidth}
Petr~A. Golovach, Daniel Lokshtanov, Saket Saurabh, and Meirav Zehavi.
\newblock Cliquewidth {III}: the odd case of graph coloring parameterized by
  cliquewidth.
\newblock In {\em Proceedings of the Twenty-Ninth Annual ACM-SIAM Symposium on
  Discrete Algorithms}, pages 262--273. SIAM, 2018.

\bibitem{GolumbicR2000}
Martin~Charles Golumbic and Udi Rotics.
\newblock On the clique-width of some perfect graph classes.
\newblock volume~11, pages 423--443. 2000.
\newblock Selected papers from the Workshop on Theoretical Aspects of Computer
  Science (WG 99), Part 1 (Ascona).

\bibitem{grotschel1981ellipsoid}
Martin Gr{\"o}tschel, L{\'a}szl{\'o} Lov{\'a}sz, and Alexander Schrijver.
\newblock The ellipsoid method and its consequences in combinatorial
  optimization.
\newblock {\em Combinatorica}, 1:169--197, 1981.

\bibitem{harary2007computational}
Frank Harary, Wolfgang Slany, and Oleg Verbitsky.
\newblock On the computational complexity of the forcing chromatic number.
\newblock {\em SIAM Journal on Computing}, 37(1):1--19, 2007.

\bibitem{hell2001fully}
Pavol Hell, Ron Shamir, and Roded Sharan.
\newblock A fully dynamic algorithm for recognizing and representing proper
  interval graphs.
\newblock {\em SIAM Journal on Computing}, 31(1):289--305, 2001.

\bibitem{hertli20143}
Timon Hertli.
\newblock 3-{SAT} faster and simpler---unique-{SAT} bounds for {PPSZ} hold in
  general.
\newblock {\em SIAM Journal on Computing}, 43(2):718--729, 2014.

\bibitem{hertli2014breaking}
Timon Hertli.
\newblock Breaking the {PPSZ} barrier for unique 3-{SAT}.
\newblock In {\em International Colloquium on Automata, Languages, and
  Programming}, pages 600--611. Springer, 2014.

\bibitem{hoos2000satlib}
Holger~H Hoos and Thomas St{\"u}tzle.
\newblock {SATLIB}: {A}n online resource for research on sat.
\newblock {\em Sat}, 2000:283--292, 2000.

\bibitem{horiyama2024theoretical}
Takashi Horiyama, Yasuaki Kobayashi, Hirotaka Ono, Kazuhisa Seto, and Ryu
  Suzuki.
\newblock Theoretical aspects of generating instances with unique solutions:
  Pre-assignment models for unique vertex cover.
\newblock In {\em Proceedings of the AAAI Conference on Artificial
  Intelligence}, volume~38, pages 20726--20734, 2024.

\bibitem{HoriyamaSS2025}
Takashi Horiyama, Kazuhisa Seto, and Ryu Suzuki.
\newblock The complexity of pre-assignment problem for unique minimum vertex
  cover on bipartite graphs.
\newblock {\em Theory Comput. Syst.}, 69(4):Paper No. 35, 21, 2025.

\bibitem{IRVING1991197}
Robert~W. Irving.
\newblock On approximating the minimum independent dominating set.
\newblock {\em Information Processing Letters}, 37(4):197--200, 1991.

\bibitem{Johansson1998}
\"Ojvind Johansson.
\newblock Clique-decomposition, {NLC}-decomposition, and modular
  decomposition---relationships and results for random graphs.
\newblock In {\em Proceedings of the {T}wenty-ninth {S}outheastern
  {I}nternational {C}onference on {C}ombinatorics, {G}raph {T}heory and
  {C}omputing ({B}oca {R}aton, {FL}, 1998)}, volume 132, pages 39--60, 1998.

\bibitem{kimura_et_al:LIPIcs.FUN.2018.25}
Kei Kimura, Takuya Kamehashi, and Toshihiro Fujito.
\newblock {The Fewest Clues Problem of Picross 3D}.
\newblock In {\em 9th International Conference on Fun with Algorithms (FUN
  2018)}, volume 100, pages 25:1--25:13, 2018.

\bibitem{KoblerR2003}
Daniel Kobler and Udi Rotics.
\newblock Edge dominating set and colorings on graphs with fixed clique-width.
\newblock {\em Discrete Appl. Math.}, 126(2-3):197--221, 2003.

\bibitem{KorhonenS2024}
Tuukka Korhonen and Marek Soko{\l}owski.
\newblock Almost-linear time parameterized algorithm for rankwidth via dynamic
  rankwidth.
\newblock In {\em S{TOC}'24---{P}roceedings of the 56th {A}nnual {ACM}
  {S}ymposium on {T}heory of {C}omputing}, pages 1538--1549. ACM, New York,
  [2024] \copyright 2024.

\bibitem{MakowskyR1999}
Johann~A. Makowsky and Udi Rotics.
\newblock On the clique-width of graph with few {$P_4$}'s.
\newblock {\em International Journal of Foundations of Computer Science},
  10(03):329--348, 1999.

\bibitem{marx2016double}
D{\'a}niel Marx and Valia Mitsou.
\newblock Double-exponential and triple-exponential bounds for choosability
  problems parameterized by treewidth.
\newblock In {\em 43rd International Colloquium on Automata, Languages, and
  Programming (ICALP 2016)}. Schloss Dagstuhl-Leibniz-Zentrum fuer Informatik,
  2016.

\bibitem{Oum2009}
Sang{-}il Oum.
\newblock Approximating rank-width and clique-width quickly.
\newblock {\em ACM Trans. Algorithms}, 5(1):Art. 10, 20, 2009.

\bibitem{OumS2006}
Sang{-}il Oum and Paul Seymour.
\newblock Approximating clique-width and branch-width.
\newblock {\em J. Combin. Theory Ser. B}, 96(4):514--528, 2006.

\bibitem{reinelt1991tsplib}
Gerhard Reinelt.
\newblock {TSPLIB}—{A} traveling salesman problem library.
\newblock {\em ORSA journal on computing}, 3(4):376--384, 1991.

\bibitem{GMXX}
Neil Robertson and Paul~D. Seymour.
\newblock Graph minors. {XX}. {W}agner's conjecture.
\newblock {\em J. Combin. Theory Ser. B}, 92(2):325--357, 2004.

\bibitem{scheder2017ppsz}
Dominik Scheder and John~P. Steinberger.
\newblock {PPSZ} for general $k$-{SAT}-making {H}ertli's analysis simpler and
  3-{SAT} faster.
\newblock In {\em 32nd Computational Complexity Conference (CCC 2017)}.
  Schloss-Dagstuhl-Leibniz Zentrum f{\"u}r Informatik, 2017.

\bibitem{THOMASON1978259}
Andrew~G. Thomason.
\newblock Hamiltonian cycles and uniquely edge colourable graphs.
\newblock In {\em Advances in Graph Theory}, volume~3 of {\em Annals of
  Discrete Mathematics}, pages 259--268. Elsevier, 1978.

\bibitem{TjusilaT2024}
Gennesaret Tjusila, Mathieu Besan\c~con, Mark Turner, and Thorsten Koch.
\newblock How many clues to give? {A} bilevel formulation for the minimum
  {S}udoku clue problem.
\newblock {\em Oper. Res. Lett.}, 54:Paper No. 107105, 6, 2024.

\bibitem{Wanke1994}
Egon Wanke.
\newblock $k$-{NLC} graphs and polynomial algorithms.
\newblock {\em Discrete Applied Mathematics}, 54(2):251--266, 1994.

\end{thebibliography}
\end{document}